\documentclass[draftclsnofoot, onecolumn, 12pt]{IEEEtran}
\usepackage{amssymb, amsmath, graphicx, paralist,subfigure}
\newtheorem{lem}{Lemma}
\newtheorem{theorem}{Theorem}
\newtheorem{defn}{Definition}

\IEEEoverridecommandlockouts

\def\mt{\mathbf{x}}
\def\mb{\mathbf}
\def\mc{\mathcal}

\begin{document}
\title{Networked estimation under information constraints$^*$ \thanks{$^*$A. Jadbabaie and U. Khan's research was supported by the following grants: ONR MURI N000140810747, NSF Career, AFOSR's Complex Networks Program and AFOSR MURI CHASE.}}
\author{Usman A. Khan$^\dagger$ and Ali Jadbabaie$^\S$
\thanks{
$^\dagger$Department of Electrical and Computer Engineering, Tufts University, {\texttt khan@ece.tufts.edu}. $^\S$Department of Electrical and Systems Engineering, University of Pennsylvania, {\texttt{jadbabai@seas.upenn.edu}}.}}

\maketitle

\begin{abstract}
In this paper, we study estimation of potentially unstable linear dynamical systems when the observations are distributed over a network. We are interested in scenarios when the information exchange among the agents is restricted. In particular, we consider that each agent can exchange information with its neighbors only once per dynamical system evolution-step. Existing work with similar information-constraints is restricted to static parameter estimation, whereas, the work on dynamical systems assumes large number of information exchange iterations between every two consecutive system evolution steps.

We show that when the agent communication network is sparely-connected, the sparsity of the network plays a key role in the stability and performance of the underlying estimation algorithm. To this end, we introduce the notion of \emph{Network Tracing Capacity} (NTC), which is defined as the largest two-norm of the system matrix that can be estimated with bounded error. Extending this to fully-connected networks or infinite information exchanges (per dynamical system evolution-step), we note that the NTC is infinite, i.e., any dynamical system can be estimated with bounded error. In short, the NTC characterizes the estimation capability of a sparse network by relating it to the evolution of the underlying dynamical system.
\end{abstract}

\section{Introduction}\label{intro}
Existing approaches to decentralized estimation or social learning are restricted to either (i) static (a fixed number) parameter estimation~\cite{degroot:74,Xiao05ascheme,Schizas08giannakis,JadSanTah10,5373900}; or, (ii) when the parameter of interest is dynamic, the estimation demands a large number of agent communications within each dynamical system evolution step\footnote{With the exception of~\cite{4739167,NBERw14040,ozd_games}, in which the parameter is a scalar and each agent is assumed to observe this scalar state. Clearly, the scalar estimation problem has each agent observable, whereas we do not assume (i) local observability; or (ii) scalar states.}~\cite{olfati:05,zamp:07,usman_tsp:07,Msechu:08}. In this paper, we consider networked estimation of (vector) dynamical systems in the context of \emph{single time-scale estimation}, i.e., when each agent can communication only once with its neighbors per dynamical system evolution step~\cite{khan_jad:cdc10}. Interested readers may also see~\cite{tam_arx}, where identity system and observation matrices are considered.

We consider social learning and networked estimation where the state of the world follows a potentially unstable discrete-time linear dynamical system. The dynamical system is monitored by a network of agents that have linear noisy observations and the agents are able to communicate over a sparse topology. Estimation within such settings heavily relies on the communication time-scale among the agents~\cite{khan_jad:cdc10}. To motivate this, we refer to Fig.~\ref{cps_ts_1}, where Fig.~\ref{cps_ts_1}(a) shows the traditional networked estimation approach, e.g., the Kalman-consensus filter~\cite{olfati:05}, where a large-number of consensus iterations are required between each~$k$ and~$k+1$;~$k$ being the time-index of the discrete-time dynamical system. On the other hand, our approach implements the agent communication and sensing at the same time-index,~$k$, of the dynamical system; we show this in Fig.~\ref{cps_ts_1}(b).
\begin{figure}
\centering
\includegraphics[width=4.5in]{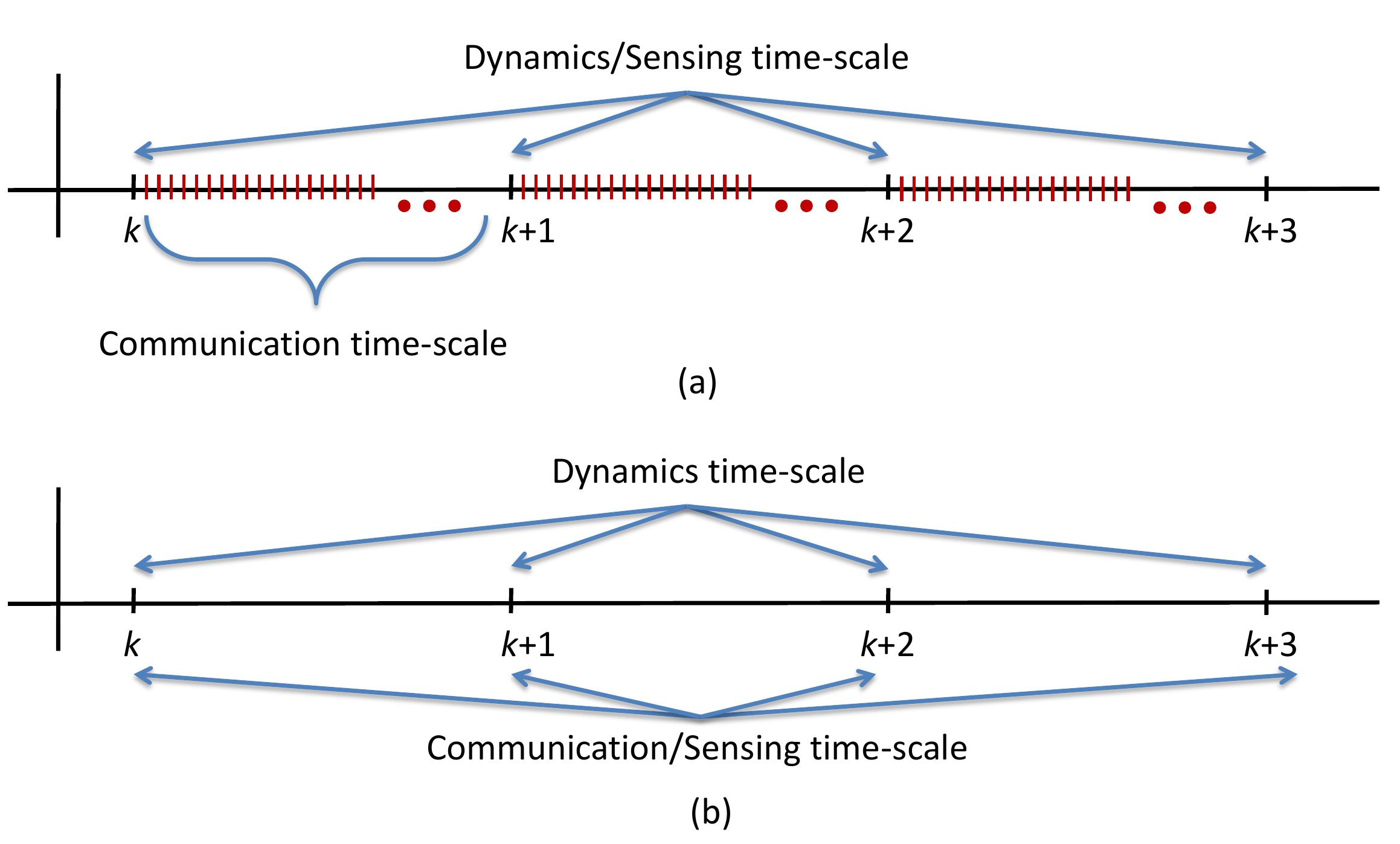}
\caption{Time-scales of dynamics, sensing and communication: (a) Large number of communication; (b) Communication and sensing at the time-scale of the dynamics.}
\label{cps_ts_1}
\end{figure}

We assume that the agents are \emph{collectively-observable}, i.e., the collection of all of the agent observations guarantees observability of the underlying state. It turns out that the collective-observability, although necessary, is not sufficient when we restrict to single time-scale estimation. This is due to the fact that the dynamics can be faster than the rate of observation fusion supported by a sparse network. The network connectivity and observation structure should be therefore related to the system instability. In this context, we address the networked estimation problem from two sides. \emph{Firstly}, we formulate the estimation problem as a spectral radius (largest eigenvalue) optimization of concerned matrices. Borrowing results from Lyapunov theory and Linear Matrix Inequalities (LMIs), we show that in the networked setting, the LMI formulation reduces to a bi-(multi-)linear optimization and may not have a solution. In cases when the LMIs do have a solution, we provide an iterative optimization procedure based on a cone complementarity linearization algorithm~\cite{rami:97}.

\emph{Secondly}, we formulate the estimation problem as a (induced) two-norm (largest singular value) optimization of the underlying matrices. We note here that the two-norm approximation of the spectral radius can be conservative; however, this design process leads to insightful networked estimation arguments that are not formulated and/or explored before. In particular, we introduce the notion of \emph{Network Tracking Capacity} (NTC), which quantifies the most unstable system that a network and a set of given observation models can track with bounded mean squared error~(MSE)\footnote{This notion can be related to the rate-constraints in information theory (the rate of sending information should be less that the channel capacity) and also to control under communication constraints (Mitter-Tatikonda~\cite{tat1:04, tat1:08}).}. We further explore the two-norm procedure with scalar-gain estimators (see structured singular value and stability margins in robust control theory~\cite{paganini_book}), and provide local design procedures using graph isomorphisms and eigenvalue bounds.

In the context of single time-scale estimators, Reference~\cite{kar-moura-ramanan-IT-2008} provides an algorithm for static parameter estimation, (see also~\cite{4749425} where a distributed sub-gradient optimization is considered). On the other hand, the distributed estimation algorithm in~\cite{4739167} considers scalar neutrally stable (unit system spectral radius) dynamical systems. Clearly, scalar dynamical systems are observable at each agent as each agent observing the scalar state is observable. This paper extends these works as we consider a multi-dimensional state-space that may not be locally observable at any strict subset of agents. Recent work on distributed estimation also includes a series of papers by Olfati-Saber et. al. (\cite{4118472,4434303,4282261,5399678,5990979}), where neighborhood (each agent plus its one-hop connected neighbors) observability is assumed. On the other hand, we do not assume observability of any agent or any strict subset of agents.

The distributed estimation algorithm we propose is related to the parallel estimators in~\cite{4287150} for closed-loop dynamics of cooperative vehicle formations. The overall framework in~\cite{4287150} is communication-oriented; the authors design the communication gain and receiver sensitivity for the parallel estimators. In addition, they also consider network topology design using a heuristic framework. There are also interesting parallels between this work and Reference~\cite{fax-murray:02}. In the special case where observation models are chosen to be identical (which is not interesting in a distributed estimation context), our proposed estimator is the dual of the controller in~\cite{4287150}. Clearly, when all of the observation models are identical, each observation model has to be observable and designing the overall system reduces to the design of a single sub-system at any agent.

The paper is organized is as follows: Preliminaries and notation are in Section~\ref{prelim}. Section~\ref{ps} states the problem and the single time-scale distributed estimation algorithm. We consider spectral radius estimator design in Section~\ref{srd} and the two-norm estimator design in Section~\ref{ntc}, whereas, Section~\ref{perf} considers error performance of the estimator. Section~\ref{sge} considers scalar gain estimators, whereas Section~\ref{ld} explores local design of the estimator parameters. Finally Section~\ref{exam} provides an illustration and Section~\ref{conc} concludes the paper.

\section{Preliminaries}\label{prelim}
This section provides preliminaries and sets notation.

\emph{System model}: Consider a time-varying parameter,~$\mt_k\in\mathbb{R}^n, n>1$, that corresponds to a phenomenon of interest. We assume the following discrete-time linear dynamical system for~$\mt_k$:
\begin{eqnarray}\label{sys1}
\mt_{k+1} = A\mt_{k} + \mb{v}_{k},
\end{eqnarray}
where~$k$ is the discrete-time index,~$A\in\mathbb{R}^{n\times n}$ is a potentially unstable system matrix, and~$\mb{v}_{k}$ is the noise in the system evolution such that
\begin{eqnarray}\label{nass1}
\mathbb{E}[\mb{v}_k] = 0,\qquad\mathbb{E}\left[\mb{v}_k\mb{v}_k^T\right] = V.
\end{eqnarray}
We assume that the phenomenon of interest~\eqref{sys1} is monitored by a network of~$N$ agents. The observation model at the~$i$th agent is given by
\begin{eqnarray}\label{sys2}
\mb{y}^i_k = H_i\mt_k + \mb{r}^i_k,
\end{eqnarray}
where the agent measurements can be collected to form a global observation, i.e.,
\begin{eqnarray}
\mb{y}_k &=& H\mb{x}_k + \mb{r}_k,\\
\triangleq\left[
\begin{array}{c}
\mb{y}_k^1\\
\vdots\\
\mb{y}_k^N
\end{array}
\right] &=&
\left[
\begin{array}{c}
H_1\\
\vdots\\
H_n
\end{array}
\right]\mb{x}_k +
\left[
\begin{array}{c}
\mb{r}_k^1\\
\vdots\\
\mb{r}_k^N
\end{array}
\right]
\end{eqnarray}
We assume
\begin{eqnarray}\label{nass2}
\mathbb{E}[\mb{r}^i_k] = 0,\qquad\mathbb{E}\left[\mb{r}^i_k\mb{r}^{iT}_k\right] = R_i,\qquad\mathbb{E}[\mb{r}_k\mb{r}_k^T]=R.
\end{eqnarray}
The noise sequences,~$\{\mb{v}_k\}_{k\geq0}$ and~$\{\mb{r}_k^i\}_{k\geq0}^{1\leq i\leq N}$, are statistically independent over time. We assume the system to be observable with the global observations, i.e., the pair~$(A,H)$ is observable. Note that any strict subset of agents may not be necessarily observable.

\emph{System instability}: We characterize the stability of the state dynamics in terms of the induced 2-norm of the system matrix,~$A$, as opposed to the spectral radius\footnote{The induced two-norm (largest singular value) and spectral radius (largest eigenvalue) are identical for normal system matrices. For non-normal system matrices, the induced two-norm is an upper bound on the spectral radius and may be a conservative estimate of stability. However, the convexity of the two-norm makes it tractable for optimization.}, i.e.,
\begin{eqnarray}\label{stab1}
a\triangleq\|A\|_2=\sqrt{\lambda_n(A^TA)},
\end{eqnarray}
where~$0\leq\lambda_1(\cdot)\leq\ldots\leq\lambda_n(\cdot)$ are the eigenvalues of the symmetric positive (semi) definite matrix~$A^TA$.

\emph{Network connectivity}: The interactions among the agents are modeled with an undirected graph,~$\mathbb{G}=(\mathcal{V},\mathcal{E})$, where~$\mathcal{V}= \{1,\ldots,N\}$ is the set of vertices and~$\mathcal{E}\subseteq \mathcal{V}\times \mathcal{V}$ is a set of ordered pairs describing the interconnections among the agents. The neighborhood at the~$i$th agent is defined as
\begin{eqnarray}\nonumber
\mathcal{N}_i\triangleq \{i\} \cup \{j~|~(i,j) \in \mathcal{E}\}.
\end{eqnarray}
The matrix~$L$ denotes the Laplacian matrix of~$\mathbb{G}$. For details on graph-theoretic concepts, see~\cite{bela_book}.

\section{Problem formulation: Single time-scale estimator}\label{ps}
In this section, we present our problem formulation. In particular, we present a single time-scale networked  estimator that is implemented at each agent to estimate the time-varying state,~$\mt_k$, of the underlying phenomenon~\eqref{sys1}. Let~$\widehat{\mb{x}}^i_{k+1}\in\mathbb{R}^n$ denote the estimate of~$\mt_k$ at agent~$i$ and time~$k+1$, given by
\begin{eqnarray}\label{est1}
\begin{array}{c}
\widehat{\mb{x}}^i_{k+1} = \\
\mbox{ }
\end{array}
\begin{array}{c}
A\left(\underbrace{\sum_{j\in\mathcal{N}_i} w_{ij}\widehat{\mb{x}}^j_k } + B_i\underbrace{\sum_{j\in\mathcal{N}_i}  H_j^T\left(\mb{y}_k^j-H_j\widehat{\mb{x}}^i_k\right)}\right),\\
\mbox{consensus update}\qquad~~\mbox{innovation update}
\end{array}
\end{eqnarray}
for~$w_{ij}\in\mathbb{R}_{\geq 0}$ such that~$\sum_{j\in\mathcal{N}_i}w_{ij}=1,\forall~j$, and~$B_i\in\mathbb{R}^{n\times n}$. At each agent, the consensus update averages the state estimates over its neighbors. On the other hand, the innovation update at agent~$i$ collects the observations at agent~$i$ and its neighbors ($j\in\mathcal{N}_i$) and uses its estimate,~$\widehat{\mb{x}}_k^i$, to form the innovation, where~$B_i$ is the local innovation gain\footnote{Static parameter estimation algorithms that are \emph{structurally similar} to~\eqref{est1} have also been considered in~\cite{4545274,sayed:arx}. In particular, \cite{sayed:arx} shows that the specific consensus and innovation structure naturally arises if the networked mean-squared cost is assumed to be the sum of local costs.}. Notice that in~\eqref{est1}, both the consensus step and the innovation step are implemented at the same time-scale, see Fig.~\ref{cps_ts_1}.

The local error process,~$\mb{e}_{k+1}^i\in\mathbb{R}^n$, at agent~$i$ and time~$k+1$ is defined as
\begin{eqnarray}\nonumber
\mb{e}_{k+1}^i \triangleq \widehat{\mb{x}}_{k+1}^i - \mt_{k+1}.
\end{eqnarray}
By concatenating the local error processes, we get the network error process,~$\mb{e}_{k+1}\in\mathbb{R}^{nN}$,
\begin{eqnarray}\label{err_pr_not}
\mb{e}_{k+1} \triangleq \left[(\mb{e}^1_{k+1})^T,\ldots,(\mb{e}^N_{k+1})^T\right]^T.
\end{eqnarray}
Let\footnote{Note that since~$w_{ij}\in\mathbb{R}_{\geq 0}$ and~$\sum_{j\in\mathcal{N}_i}w_{ij}=1,\forall~j$,~$W$ is a stochastic matrix.}~$W \triangleq \{w_{ij}\}$,~$B \triangleq \mbox{diag}[B_1,\ldots,B_N]$, and
\begin{eqnarray}\label{DH}
D_H \triangleq
\left[
\begin{array}{ccc}
\sum_{j\in\mathcal{N}_1} H_j^TH_j&&\\
&\ddots&\\
&&\sum_{j\in\mathcal{N}_N} H_j^TH_j
\end{array}
\right].
\end{eqnarray}
It can be verified that the network error process,~$\mb{e}_{k+1}$, is given by
\begin{eqnarray}\label{err_eq3}
\mb{e}_{k+1} = P\mb{e}_k + \mb{u}_k,
\end{eqnarray}
where $P\triangleq(I_N\otimes A)(W \otimes I_n - BD_H)$, $\mb{u}_k \triangleq \mb{\phi}_k - \mb{1}_N\otimes \mb{v}_k$, and
\begin{eqnarray}\nonumber
\mb{\phi}_k &\triangleq&(I_N\otimes A)B\left[
\begin{array}{c}
\sum_{j\in\mathcal{N}_1} H_j^T\mb{r}_k^j\\
\vdots\\
\sum_{j\in\mathcal{N}_N} H_j^T\mb{r}_k^j
\end{array}
\right].
\end{eqnarray}

The design and analysis of the single time-scale networked estimator~\eqref{est1} concerns with studying the stability and steady-state performance of the networked error process~\eqref{err_eq3}. We take the following approach in this regard:
\begin{enumerate}[(i)]
\item Estimator design using LMIs (Section~\ref{srd}): We study the design of the matrices~$W$ and~$B$ such that~$\rho(P)<1$ (stable estimation error) using an LMI-based approach. Although this approach requires minimal assumptions, it does not guarantee the existence of~$W$ and~$B$ with~$\rho(P)<1$. This is because the underlying structure on~$W$ (graph sparsity) and~$B$ (block-diagonal) restricts the solution of the corresponding Linear Matrix Inequality (LMI).
\item Two-norm design (Section~\ref{ntc}): We use the fact that~$\rho(P)<\|P\|_2$ and design the matrices~$W$ and~$B$ such that~$\|P\|_2<1$, i.e., instead of using the spectral radius constraint for stability, we resort to a two-norm constraint. This relaxation leads to the Network tracking Capacity (NTC) -the most unstable dynamical system that can be estimated with bounded estimation error under the structural constraints on~$W$ and~$B$- arguments.
\item Performance (Section~\ref{perf}): We study the performance of the networked estimator in terms of upper bounds on the steady-state error covariance.
\item Scalar gain estimators (Section~\ref{sge}): We study the scalar gain estimator by restricting~$W=I_N-\alpha L$ and~$B=\alpha I_{nN}$. Hence, instead of designing the entire matrices,~$W$ and~$B$, we only consider a scalar parameter,~$\alpha\in\mathbb{R}$, leading to local design of~$\alpha$ at each agent.
\item Local design (Section~\ref{ld}): We consider local design of the scalar-gain estimator parameters, where each agent can locally choose the scalar design parameter,~$\alpha$; we achieve this by using graph isomorphisms and some eigenvalue bounds.
\end{enumerate}

\subsection{Special cases}
Before we proceed with the rest of the paper, we illustrate the applicability of the proposed estimator in~\eqref{est1} to some special cases.

\subsubsection{Scalar systems with no collaboration} For scalar dynamical systems ($n=1$) with no collaboration (i.e., without the consensus update,~$w_{ii}=1$ and~$w_{ij}=0, i\neq j$), the error (at the~$i$th agent) in the proposed estimator is given by
\begin{eqnarray}\nonumber
\widehat{e}_{k+1}^i &=& \underbrace{ax_k + v_k}_{x_{k+1}} - a \left(\widehat{x}_k^i + b_i(k)h_i(h_ix_k+r_k^i - h_i\widehat{x}_k^i)\right),\\\nonumber
&=& a\left(1 - b_i(k)h_i^2\right)e_k^i + v_k - ab_i(k)h_ir_k^i.
\end{eqnarray}
With no collaboration (no interagent interaction), the above error dynamics lead to a bounded MSE by choosing~$b_i(k)=h_i^{-2}$ and the error dynamics become
\begin{eqnarray}
\widehat{e}_{k+1}^i &=& v_k - ah_i^{-1}r_k^i,
\end{eqnarray}
regardless of the value of~$a$. In particular, assuming independent system and observation noise and~$E[v_kv_k^T] = \sigma_v^2,~\mathbb{E}[r_k^ir_k^{iT}]=\sigma_r^2$, we have
\begin{eqnarray}
\mathbb{E}[e_k^ie_k^{iT}] &=& \sigma_v^2 + ah_i^{-2} \sigma_r^2.
\end{eqnarray}
This is because the scalar system is observable at each agent~$i$ as long as~$h_i\neq 0$. Hence, collaboration is not required to guarantee bounded MSE and the proposed estimator results in bounded MSE for any\footnote{We note that non-trivial choices of the estimator gain,~$b_i(k)$, different from~$h_i^{-2}$, exist that also minimize the steady-state error at the price of slower convergence. With the choice of~$h_i^{-2}$, we note that the convergence is achieved in one-step; choosing~$b_i=h_i^{-2}$ is sufficient to show that a bounded MSE does not require collaboration among the agents in the case of scalar dynamical systems.}~$a$. However, with the addition of inter-agent information exchange, the steady performance of the estimator can be improved as we show next.

\subsubsection{Scalar systems with collaboration}\label{sc:nss}
When the agents interact over a sparse communication network, it can be shown that the error at the~$i$th agent is given by
\begin{eqnarray}
e_{k+1}^i
&=& a\left(\sum_{j\in\mc{N}_i}w_{ij} e_k^j - b_i(k)\sum_{j\in\mc{N}_i}h_j^2e_k^i\right) + v_k - ab_i\sum_{j\in\mc{N}_i}h_jr_k^j,
\end{eqnarray}
where we used~$\sum_{j\in\mc{N}_i}w_{ij}=1$. Due to the interaction, the error at each any agent is coupled with the neighboring error dynamics (and subsequently to errors at all of the nodes in the network assuming a connected network.) From the previous discussion, the networked error process as \begin{eqnarray}
\mb{e}_{k+1} = a\left(W-\beta(k)I_N\right)\mb{e}_k + \mbox{blockdiag}\left\{v_k - a\beta(k)r_k^j\right\},
\end{eqnarray}
where we have chosen~$b_i(k) = \beta(k)(\sum_{j\in\mc{N}_i}h_j^{2})^{-1}$; the inverse exists if~$h_i\neq0,~\forall i$. Note that all of the possible choices of~$\beta(k)$ such that~$\rho(a\left(W-\beta(k)I_N\right))<1$, result into bounded MSE of the networked estimator. The estimator design, thus, depends on the weights chosen in~$W$ (stochastic) and~$\beta(k)$. The error dynamics further highlight that there exist unstable dynamics ($a>1$) such that our proposed estimator results into bounded MSE. The unstable dynamics result from the stability margin in~$a\left(W-\beta(k)I_N\right)$.

In a further special case of neutrally stable system, i.e.,~$a=1$, Reference~\cite{4739167} shows that~$\beta(k)$ can be chosen such that~$\beta(k)$ follow some persistence and diminishing conditions. In other words,~$\beta(k)$ is less than the (some) weight chosen in~$W$, but they sum to infinity\footnote{Such choices are motivated by standard stochastic approximation techniques in recursive algorithms~\cite{Kushner,Kushner-Yin,kush_book}. These algorithms have been explored for non-scalar static systems in the context of distributed localization and consensus, for details see~\cite{karrandomtopologynoise,usman_loctsp:08,usman_slretsp:09}.}. The advantage of choosing a diminishing gain,~$\beta(k)\rightarrow 0$, is that the constant in the networked error expression reduces to~$v_k$. Using this observation combined with an assumption of \emph{diminishing innovations}, i.e.,~$v_k\rightarrow 0$,~\cite{4739167} shows that~$\mb{e}_k\rightarrow \mb{0}$. However, an avid reader may note that the persistence assumptions on~$\beta(k)$ can only be made when~$a\leq 1$. In other words, for a non-trivial system ($a>1$), we cannot choose~$\beta(k)\rightarrow 0$.

\subsubsection{Social networks--Static parameter estimation}
Static parameter estimation ($A=I, v_k=0$) and aggregation behavior has been studied extensively in social networking applications in both Bayesian~\cite{sn3,sn4} and non-Bayesian settings~\cite{sn1,sn2,JadSanTah10}, for more details see references therein. In particular,~\cite{JadSanTah10} considers a scalar parameter to be estimated using belief propagation adding a network (consensus) term to the DeGroot model~\cite{degroot:74}. In an abstract setting, Reference~\cite{kar-moura-ramanan-IT-2008} discusses a single time-scale estimator for static vector parameters that is close to our formulation (when restricted to the static case) in~\eqref{est1}.

\section{Estimator design using LMIs}\label{srd}
In this section, we consider the estimator design using an LMI based approach. Recall that we would like to choose~$W$ and~$B$ such that (see~\eqref{err_eq3})
\begin{eqnarray}\label{rho_des}
\rho(W\otimes A - (I_N\otimes A)BD_H) < 1.
\end{eqnarray}
The spectral radius is a non-convex function (unless the matrix is normal that may not be true in our formulation). Due to this, a tractable optimization is not not possible to \emph{minimize}~$\rho(\cdot)$. However, assuming a full gain matrix,~$B$, one may implement an LMI-based design that satisfies~\eqref{rho_des} (this design may not minimize the spectral radius) as long as the pair~$(W\otimes A, D_H)$ is observable; see~\cite{lmi_book},\cite{paganini_book} for details. We note here that in our formulation the gain,~$B$, has structural constraints (block-diagonal); the corresponding LMIs, in general, do not have a solution. Clearly, this is the main difficulty in distributed estimation and control as convex/semidefinite approaches are not directly applicable. To this end, we resort to an iterative procedure to solve LMIs under structural constraints.

Before we proceed with our methodology, we note that this formulation assumes the observability of~$(W\otimes A,D_H)$. Note that even when we assume~$(A,H)$-observability, we still have to choose an appropriate~$W$ that ensures the observability of~$(W\otimes A,D_H)$ in order to use LMIs. To this end, References~\cite{khan_mrd:asil11,khan_jad:cdc11,khan_mrd:camsap11}, discuss agent connectivity protocols to ensure \emph{structured observability} of~$(W\otimes A,D_H)$ given the observability of~$(A,H)$. These protocols are independent of a particular fusion rule, i.e., a strict choice of elements in~$W$ is not enforced as long as the sparsity (zero/non-zero pattern) is fixed\footnote{In fact, it can be shown that if~$(W\otimes A,D_H)$ is observable for \emph{some} choice of~$W$, then it is observable for almost all~$W$'s with the same sparsity. The set of~$W$'s where it is not observable has zero Lebesgue measure~\cite{woude:03}.}. In the following, we assume that~$(W\otimes A,D_H)$ is observable by designing an appropriate communication network among the agents as discussed in~\cite{khan_mrd:asil11,khan_jad:cdc11,khan_mrd:camsap11}.

\subsection{LMIs under structural constraints}
Our method uses a cone complementarity linearization algorithm provided in~\cite{rami:97} (see Wireless Control Network (WCN) in~\cite{5717159}). We describe the procedure below. Define
\begin{eqnarray}
\widehat{A} \triangleq (W\otimes A - (I_N\otimes A)BD_H).
\end{eqnarray}
From Lyapunov theory~\cite{paganini_book}, it is well-known that~$\rho(\widehat{A}) < 1$ when the following Linear Matrix Inequality (LMI) holds for some~$X\succ 0$ (`$\succ$' denotes positive-definiteness),
\begin{eqnarray}\nonumber
X &-& \widehat{A}^TX\widehat{A} \succ 0,
\end{eqnarray}
or, equivalently, when
\begin{eqnarray}\nonumber
\left[
\begin{array}{cc}
X& \widehat{A}^TX,\\
X\widehat{A} &X
\end{array}
\right] \succ 0,
\end{eqnarray}
with~$X\succ 0$, by using the Schur complement. Since the above is non-linear in the design parameter (products of~$X$,$W$ and~$X$,$B$), we perform a change of variable and note that~$\rho(\widehat{A})<1$, if and only if there exists~$X,Y\succ 0$ such that
\begin{eqnarray}\label{lin_lmi}
\left[
\begin{array}{cc}
X& \widehat{A}^T,\\
\widehat{A}& Y
\end{array}
\right] \succ 0,
\end{eqnarray}
with~$X=Y^{-1}$. The above LMI~\eqref{lin_lmi} is linear in the design parameters~$X, W$, and~$B$, but the corresponding constraint,~$X=Y^{-1}$, is non-convex.

Here, we use the approach in~\cite{rami:97} to approximate~$X=Y^{-1}$ with a linear function. In particular, the matrices,~$X,Y\succ 0$, satisfy~$X=Y^{-1}$, if and only if they are optimal points of the following optimization problem.
\begin{eqnarray}\nonumber
\min~tr(XY) \mbox{ subject to }
\left[
\begin{array}{cc}
X& I\\
I& Y
\end{array}
\right]\succeq 0,
\end{eqnarray}
with~$X,Y\succ 0$. Furthermore,~$W$ has to follow the graph,~$\mathbb{G}$, sparsity, and~$B$ is block-diagonal. The above discussion can be summarized in the following lemma.
\begin{lem}
The networked estimator is stable when structured matrices,~$W$ and~$B$, are the solution of the following optimization.
\begin{eqnarray}\label{min_op}
\mbox{min } tr(XY),&&\\\nonumber
\left[
\begin{array}{cc}
X& \widehat{A}^T,\\\nonumber
\widehat{A}& Y
\end{array}
\right] \succ 0,&&
\left[
\begin{array}{cc}
X& I,\\\nonumber
I& Y
\end{array}
\right]\succeq 0,\\\nonumber
B\mbox{ is block-diagonal},&&\\\nonumber
W\mbox{ is stochastic and }&& W\sim\mathbb{G},\\\nonumber
X,Y&\succ& 0.
\end{eqnarray}
\end{lem}
The objective function is again a minimization of a non-convex object. (Notice that since the second LMI is equivalent to~$X=Y^{-1}$~\cite{lmi_book}, the minimum trace is achieved at~$X=Y^{-1}$ and the optimal value is~$nN$.) To this end, we replace the trace operator over the product of~$X$ and~$Y$ with a linear approximation~\cite{pang:95,rami:97},
\begin{eqnarray}
\phi_{\mbox{lin}} (X,Y) = tr (Y_0 X + X_0 S),
\end{eqnarray}
and an iterative algorithm can be used to minimize~$tr(XY)$, under the constraints on~$W$ and~$B$. The iterative algorithm~\cite{rami:97} is as follows:
\begin{enumerate}[(i)]
\item Find feasible points~$X_0, Y_0, W, B$. If no such points exist, Terminate.

\item Find~$X_{t+1},Y_{t+1}$ by minimizing~$tr (Y_t X + X_t Y)$ under the constraints in~\eqref{min_op}.

\item Terminate when~$\rho(\widehat{A})<1$ or according to a desirable stopping criterion.
\end{enumerate}

\subsection{Discussion}
In the following, we briefly review the LMI-based spectral radius design and discuss its limitations.

\begin{inparaenum}[(i)]
\item Let~$s_{t+1} = tr(Y_t X_{t+1}) + X_{t+1} Y)$, then it is shown in~\cite{rami:97} that~$s_t$ is a decreasing sequence that converges to~$2nN$; the convergence to~$2nN$ is because as~$t\uparrow$,~$X_{t+1}\rightarrow Y^{-1}$ and~$X\in\mathbb{R}^{nN\times nN}$. However, there is no analytical characterization of the convergence rate.

\item An alternative stopping criterion, instead of (iii), can be established in terms of reaching within~$2nN + \varepsilon$ of the trace objective.

\item The iterative procedure given above, similar to the cone-complementarity linearization algorithm in~\cite{rami:97}, is a centralized algorithm and has to be implemented at a center. However, the center has to implement this process only once, off-line, and then it may broadcast the appropriate estimator gain to each agent. Afterwards, the center plays no role in the implementation of local estimators at each agent; each agent, subsequently, observes and performs \emph{in-network} operations to implement the estimator.

\item A single time-scale algorithm can also be implemented, where the above iterative procedure is implemented at the same time-scale~$k$ as of the dynamical system in~\eqref{sys1}. With this approach, the estimator gain becomes a function of~$k$ (i.e.,~$K_{k+1}$) and may be transmitted to each agent at each time-step~$k$. This is helpful when the implementation is assumed in real-time.

\item The LMI based spectral radius design provided in this section guarantees a solution if feasible points point,~$X_0, Y_0, W, B$, exist that satisfy the constraints of the minimization in~\eqref{min_op}. If no such points can be obtained then the spectral radius design does not result in a solution. Hence, when a center-based implementation is not feasible or a solution does not exist using the LMI-based methods, we resort to the two-norm design procedure described in the next section.
\end{inparaenum}

\section{Estimator design under a convex relaxation}\label{ntc}
As we explained before, the LMI-based design is centralized and may not result in a stable error process. To address such cases, we use a convex relaxation, $\rho(P)\leq\|P\|_2$, and provide a design that is based on~$\|P\|_2<1$ to ensure a stable error process. To this end, we characterize the Network Tracking Capacity (NTC) that quantifies the most unstable dynamical system, which a network and a set of given observation models may track with bounded error. With the help of the previous discussion, we have the following definition.

\begin{defn}\label{ntc_def}[Network tracking capacity,~$C$] Given the network connectivity, i.e., the network communication graph,~$\mathbb{G}$, and the observation matrices,~$\{H_i\}_{1\leq i\leq N}$, the network tracking capacity,~$C$, is defined as the most unstable dynamical system (in the~$2$-norm sense) that can be estimated by~\eqref{est1} with bounded mean-squared error (MSE).
\end{defn}

In the following, we derive an expression for the NTC,~$C$, and show that~$\forall~a\triangleq\|A\|_2<C$, there are a set of choices for the~$W=\{w_{ij}\}$ and~$B=\mbox{blockdiag}\{B_i\}$ such that the networked estimator in~\eqref{est1} results into bounded estimation error. In Section~\ref{sge}, we consider a special case ($w_{ii} = 1-\alpha,w_{ij,j\neq i}=\alpha, B_i=\alpha I_n~$) and design the parameter~$\alpha$, which leads to local design of the networked estimator studied in Section~\ref{ld}.

\subsection{Network tracking capacity}
In the following, we derive a mathematical expression for the NTC and explore some of its properties. To this end, we note the following:
\begin{eqnarray}\nonumber
\|P\|_2 &=& \|(I_N\otimes A)(W \otimes I_n - BD_H)\|_2,\\\nonumber
&\leq& \|I_N\otimes A\|_2\|W \otimes I_n - BD_H)\|_2,\\\nonumber
&=& a\|W \otimes I_n - BD_H\|_2.
\end{eqnarray}
Note that all such stochastic matrices~$W\in\mathbb{R}^{N\times N}$ and gain matrices~$B_i\in\mathbb{R}^{n\times n}$, which guarantee~$\|W \otimes I_n - BD_H\|_2<1/a$, ensure a stable error process, i.e.,~$\|P\|_2<1$. Clearly, feasible~$W,B$ may not exist for all dynamical systems (i.e., for arbitrarily large~$a$). Hence, it is natural to ask what is the range of~$a$ for which we can guarantee~$\|W \otimes I_n - BD_H\|_2<1/a$ with appropriate choices of~$W$ and~$B$. This exposition naturally lends itself to the Definition~\ref{ntc_def} of Network Tracking Capacity (NTC) we introduced earlier. The following theorem mathematically characterizes the NTC,~$C$.

\begin{theorem}
The network tracking capacity,~$C$, is given by\footnote{Note the parallels between this definition and the structured singular value (cf.~\cite{paganini_book} and the references therein). The NTC is essentially quantifying the margin of stability of the estimator dynamics.}
\begin{eqnarray}\label{ntc_ex}
C = \dfrac{1}{\min_{W,B}\|W \otimes I_n - BD_H\|_2}.
\end{eqnarray}
\end{theorem}
\begin{proof}
To prove the above theorem, we have to show that for any~$a<C$, there exist a stochastic matrix~$W$ and gain matrices~$B_i$'s such that we have a stable error process, i.e.,~$\|P\|_2<1$. We assume~$a<C$ and choose
\begin{eqnarray}\nonumber
&&[W,B] = \mbox{argmin}_{W^\prime,B^\prime}\|W^\prime \otimes I_n - B^\prime D_H\|_2,\\\label{alpha_opt}
&&\mbox{subject to}\qquad W^\prime \mbox{~is stochastic},
\end{eqnarray}
then we have
\begin{eqnarray}\nonumber
\|P\|_2 &\leq& a\|W \otimes I_n - BD_H\|_2,\\\nonumber
&<& \dfrac{1~~(\min_{W,B}\|W \otimes I_n - BD_H\|_2)}{\min_{W,B}\|W \otimes I_n - BD_H\|_2},\\\nonumber
&<& 1.
\end{eqnarray}
\end{proof}
The above theorem shows that for all~$a<C$, there exist~$W, B$ as given in~\eqref{alpha_opt} that ensure a stable error process\footnote{This capacity argument is similar to the information-theoretic capacity argument, where the rate of information has to be less than the channel capacity for reliable communication. Achieving the capacity in either case requires careful design of the underlying parameters, i.e.,~$W,B$ in our case, under the appropriate constraints.}, i.e.,~$\|P\|_2<1$. Some properties of the NTC are explored below.
\begin{lem}\label{lem_ntc_ran}
The NTC lies in the interval~$[1,\infty]$,~i.e., $1\leq C\leq \infty.$
\end{lem}
\begin{proof}
Choosing~$W=I_N$ and~$B=0_{nN}$ provides the lower bound. To get the upper bound consider a communication graph and an observation model that give~$BD_H=W\otimes I_n$. For example, an observation model where each neighborhood is \emph{one-step} observable (i.e.,~$(\sum_{j\in\mathcal{N}_i}H_j^TH_j)^{-1}$ exists~$\forall~i$) suffices by choosing~$W=I$ and~$B_i = (\sum_{j\in\mathcal{N}_i}H_j^TH_j)^{-1}$.
\end{proof}
The above lemma shows that the NTC for a system where each neighborhood is one-step observable results into an infinite tracking capacity, i.e., any dynamical system can be tracked with bounded error. Similar argument shows that the capacity for a one-step observable system with a fully-connected agent network is also infinite (as~$\sum_{j\in\mathcal{V}}H_j^TH_j$ is always invertible). In the case where we have no observations, i.e.,~$D_H=0_{nN}$, it can be verified that~$C=1$, i.e., only stable dynamical systems ($a<1$) can be tracked with bounded error.

\section{Performance}\label{perf}
In this section, we study the performance of the single time-scale estimator~\eqref{est1}. For this purpose, our main concern is the noise process,~$\mb{u}_k$ (linear combination of the system and the observation noise), in the error process,~$\mb{e}_k$, see~\eqref{err_eq3}. Note that
\begin{eqnarray}\label{eta_mean}
\mathbb{E}[\mb{u}_k] = \mb{0},\qquad\mathbb{E}\left[\mb{u}_k\mb{u}_j^T\right] = \mb{0},\qquad k\neq j,
\end{eqnarray}
since both~$\mb{v}_k$ and~$\mb{\phi}_k$ are zero-mean and statistically independent over time. Define~$\Sigma$ to be the noise covariance matrix of the error process, i.e.,
\begin{eqnarray}\nonumber
\Sigma&\triangleq& \mathbb{E}\left[\mb{u}_k\mb{u}_k^T\right]= \Phi + \mb{1}_N\mb{1}_N^T\otimes V,
\end{eqnarray}
where
\begin{eqnarray}\nonumber
\Phi&\triangleq&\mathbb{E}[\mb{\phi}_k\mb{\phi}_k^T]= (I_N\otimes A) B (\mc{A}\otimes I_n)\overline{R}(\mc{A}\otimes I_n)^TB^T(I_N\otimes A)^T,
\end{eqnarray}
with
\begin{eqnarray}
\overline{R} &=& \left[
\begin{array}{ccc}
H_1^TR_1H_1&&\\
&\ddots&\\
&&H_N^TR_NH_N
\end{array}
\right],\\
\mc{A} &=& \mbox{Adj}(\mb{G}) + I_N.
\end{eqnarray}
Using standard stability arguments for linear systems~\cite{paganini_book}, it can be shown that the error process in~\eqref{err_eq3} is stable, and asymptotically unbiased,~i.e., $\lim_{k\rightarrow\infty} \mathbb{E}[\mb{e}_{k+1}] = \mb{0}$, when we have
\begin{eqnarray}\label{Pnorm1}
p\triangleq\|P\|_2 &<& 1.
\end{eqnarray}
Furthermore, it can be verified that
\begin{eqnarray}\nonumber
S_{k+1}&\triangleq&\mathbb{E}\left(\mb{e}_{k+1}\mb{e}_{k+1}^T\right)= P^{k+1}S_0(P^{T})^{k+1} + \sum_{j=0}^k P^j \Sigma (P^T)^j,
\end{eqnarray}
and, asymptotically, assuming~$\|P\|_2<1$, we have
\begin{eqnarray}\nonumber
S_\infty&\triangleq&\lim_{k\rightarrow\infty}S_{k+1}= \sum_{j=0}^\infty P^j \Sigma (P^T)^j.
\end{eqnarray}
The two-norm of the steady-state error covariance can now be bounded above as
\begin{eqnarray}\nonumber
\|S_\infty\|_2 &\leq& \sum_{j=0}^\infty \|P^j \Sigma (P^T)^j\|_2,\\
&\leq& \dfrac{\|\Sigma\|_2}{1-p^2}.
\end{eqnarray}
The above upper bound can be further expanded as
\begin{eqnarray}\nonumber
\|S_\infty\|_2 &\leq& \dfrac{\|\mb{1}_N\mb{1}_N^T\otimes V + (I_N\otimes A) B (\mc{A}\otimes I_n)\overline{R}(\mc{A}\otimes I_n)^TB^T(I_N\otimes A)^T\|_2}{1 - \|(I_N\otimes A)(W \otimes I_n - BD_H)\|_2^2},\\
&\leq& \dfrac{N\|V\|_2 + \|(I_N\otimes A) B (\mc{A}\otimes I_n)\overline{R}(\mc{A}\otimes I_n)^TB^T(I_N\otimes A)^T\|_2}{1 - a^2\|W \otimes I_n - BD_H\|_2^2},
\end{eqnarray}
where we have employed the triangle inequality and the fact that
\begin{eqnarray}\nonumber
\|\mb{1}_N\mb{1}_N^T\otimes V\|_2 = \|\mb{1}_N\mb{1}_N^T\|_2\|V\|_2 = N\|V\|_2.
\end{eqnarray}
Similarly, using sub-multiplicative property of the two-norm and
\begin{eqnarray}\nonumber
\|\mc{A}\otimes I_n\|_2 = \|\mc{A}\|_2 \leq \|\mb{1}_N\mb{1}_N^T\|_2 = N,
\end{eqnarray}
we can further simplify the upper bound on~$\|S_\infty\|$ as
\begin{eqnarray}
\dfrac{1}{N}\|S_\infty\|_2 &\leq& \dfrac{\|V\|_2 + a^2N \|B\|_2^2 \|\overline{R}\|_2}{1 - a^2\|W \otimes I_n - BD_H\|_2^2},
\end{eqnarray}
where we scale the networked steady-state error by~$1/N$ to write the steady error at each agent. Clearly, when we choose~$W$ and~$B$ such that~$a<1/\|W\otimes I_n-BD_H\|$, we have~$a\|W\otimes I_n-BD_H\|<1$ and the denominator is never~$0$. In other words, when~$a<C$, the steady-state error is bounded. In addition, the farther we operate from the capacity, the lower the steady-state error bound.

Notice that the upper bound on the steady state error is small when both~$\|B\|_2$ and~$\|W \otimes I_n - BD_H\|_2$ are small. Hence, we may implement the following convex optimization.
\begin{eqnarray}\label{co_perf}
\min_{W,B} \|B\|_2 &+& \|W \otimes I_n - BD_H\|_2,\\\nonumber
\mbox{subject to}&&\|W \otimes I_n - BD_H\|_2 < \dfrac{1}{a},\\\nonumber
&&W\mbox{ is stochastic} \mbox{ and } W\sim \mathbb{G},\\\nonumber
&&B\mbox{ is block-diagonal}.
\end{eqnarray}
Clearly, when the above constraints are satisfied, a stable estimator exists under the structure on the weight matrix,~$W$, and the gain matrix,~$B$. Minimizing the objective under stability and structural constraints further ensures a performance limit on the steady state error.

\section{Scalar gain estimators}\label{sge}
Consider a special case of the estimator in~\eqref{est1} by choosing~$W=I_N - \alpha L$ and~$B_i = \alpha I_n,\forall~i$, for some~$\alpha\in\mathbb{R}_{\geq 0}$. The resulting estimator at agent~$i$ and time~$k$ is given by
\begin{eqnarray}\label{est2}
\widehat{\mb{x}}^i_{k+1} = A\widehat{\mb{x}}^i_k - \alpha A\sum_{j\in\mathcal{N}_i}\left(\widehat{\mb{x}}^i_k - \widehat{\mb{x}}^j_k  - H_j^T\left(\mb{y}_k^j-H_j\widehat{\mb{x}}^i_k\right)\right).
\end{eqnarray}
We term this estimator as the \emph{scalar gain estimator}. We denote the NTC for this estimator by~$C_\alpha$. Clearly, we have
\begin{eqnarray}\label{CalC}
C_\alpha\leq C.
\end{eqnarray}
It can be verified that the matrix~$P$ in~\eqref{err_eq3} for scalar gain estimators is given by
\begin{eqnarray}
P&\triangleq&(I_N\otimes A)(I_{nN} - \alpha Q),
\end{eqnarray}
where
\begin{eqnarray}
Q = L\otimes I_N + D_H.
\end{eqnarray}
The rest of this section is dedicated to the study of scalar gain estimators. To establish our results, we provide the following lemma.
\begin{lem}\label{lem_ntc_prop}
We have
\begin{eqnarray}\nonumber
\min_{\alpha}\|I_{nN} - \alpha Q\|_2 &=& \dfrac{\lambda_{nN}(Q)-\lambda_{1}(Q)}{\lambda_{nN}(Q)+\lambda_{1}(Q)},\\\label{alpha_opt2}
\alpha_{\mbox{\scriptsize opt}}\triangleq\mbox{argmin}_{\alpha}\|I_{nN} - \alpha Q\|_2 &=& \dfrac{2}{\lambda_{nN}(Q)+\lambda_{1}(Q)}.
\end{eqnarray}
\end{lem}
\begin{proof}
Since~$Q$ is symmetric positive semi-definite, its eigenvalues are positive reals. Since~$\|I_{nN} - \alpha Q\|_2$ is also symmetric, we have
\begin{eqnarray}\nonumber
\|I-\alpha Q\|_2 = \max_{1\leq i\leq nN}|\lambda_i(I-\alpha Q)|.
\end{eqnarray}
The eigenvalues of~$I-\alpha Q$ are~$1-\alpha\lambda_i(Q)$. Hence,
\begin{eqnarray}\nonumber
\|I-\alpha Q\|_2 = \max\{1 - \alpha\lambda_{1}(Q),\alpha\lambda_{nN}(Q) - 1\}.
\end{eqnarray}
As a function of~$\alpha$, both~$1 - \alpha\lambda_{1}(Q)$ and~$\alpha\lambda_{nN}(Q) - 1$ are straight lines with slopes~$-\lambda_1(Q)$ and~$\lambda_{nN}(Q)$, respectively. Hence, the lines~$1 - \alpha\lambda_{1}(Q)$ and~$\alpha\lambda_{nN}(Q) - 1$ intersect at~$\alpha_{\mbox{\scriptsize int}}$ (having slopes opposite in sign). The point of intersection is given by
\begin{eqnarray}\nonumber
1 - \alpha_{\mbox{\scriptsize int}}\lambda_{1}(Q)&=&\alpha_{\mbox{\scriptsize int}}\lambda_{nN}(Q) - 1,\\
\Rightarrow \alpha_{\mbox{\scriptsize int}} &=& \dfrac{2}{\lambda_{nN}(Q)+\lambda_{1}(Q)}.
\end{eqnarray}
It can be verified that the considered minimization lies at~$\alpha_{\mbox{\scriptsize int}}$, i.e., $\alpha_{\mbox{\scriptsize opt}} = \alpha_{\mbox{\scriptsize int}},$ and is thus
\begin{eqnarray}\nonumber
\min_\alpha\|I-\alpha Q\|_2 &=& \min_\alpha\max\{1 - \alpha\lambda_{1}(Q),\alpha\lambda_{nN}(Q) - 1\},\\
&=& 1 - \alpha_{\mbox{\scriptsize opt}}\lambda_{1}(Q),
\end{eqnarray}
and the lemma follows.
\end{proof}
From the above lemma, we note that an equivalent expression for the NTC,~$C_\alpha$, is
\begin{eqnarray}\label{ntc_exp2}
C_\alpha = \dfrac{\lambda_{nN}(Q)+\lambda_{1}(Q)}{\lambda_{nN}(Q)-\lambda_{1}(Q)}.
\end{eqnarray}

The following lemma establishes the NTC for connected (agent network) systems with one-step collective observability\footnote{Notice that for scalar-gain estimators, since we have a single parameter,~$\alpha$, to design both network weights ($W$) and estimator gain ($B$), we require stronger (one-step) observability. Recall that the standard~$(A,H)$-observability is~$n$-step. However, as we motivated before, no strict subset agents is assumed to be one-step observable.}, i.e., the following matrix
\begin{eqnarray}
G\triangleq \sum_{j\in\mathcal{V}}H_j^T H_j
\end{eqnarray}
is invertible.

\begin{lem}\label{lem_sc:nss}
Let~$\mathbb{G}$ be connected, i.e.,
\begin{eqnarray}\nonumber
0=\lambda_1(L)<\lambda_2(L)\leq\ldots\leq\lambda_N(L),
\end{eqnarray}
and let the observation models be collectively-observable in \emph{one time-step}. Then~$C_\alpha>1$.
\end{lem}
\begin{proof}
To prove the above lemma, we first show that, for a connected-observable system, the matrix~$Q$ is strictly positive-definite. We note that, since the graph is connected, the only eigenvector of the graph Laplacian corresponding to the~$0$ eigenvalue is~$\mb{1}_N$~\cite{bela_book}. Hence, the eigenvector of~$L\otimes I_n$ corresponding to the~$0$ eigenvalue is~$\mb{1}_N\otimes \mathbf{a}$, for any~$\mb{a}\in\mathbb{R}^n$. For this eigenvector
\begin{eqnarray}\nonumber
(\mb{1}_N\otimes \mathbf{a})^TD_H(\mb{1}_N\otimes \mathbf{a}) &=& \sum_{i=1}^N\mathbf{a}^T\sum_{j\in\mathcal{N}_i}H_j^T H_j\mb{a},\\\nonumber &\geq& \mathbf{a}^T G \mb{a} > 0,
\end{eqnarray}
since~$G$ is invertible and strictly positive definite. Clearly, for any vector,~$\mb{b}\neq\mb{1}_N\otimes \mathbf{a}$, for~$\mb{a}\in\mathbb{R}^n$, we have both~$\mb{b}^T(L\otimes I_n) \mb{b}>0$ and~$\mb{b}^TD_H\mb{b}>0$. Hence,
\begin{eqnarray}\nonumber
\mathbf{b^\prime}^TQ\mathbf{b^\prime} > 0, \qquad \forall~\mb{b^\prime}\in\mathbb{R}^{nN}.
\end{eqnarray}
Thus, the eigenvalues of~$Q$ (for connected-observable systems) are strictly positive, i.e.,~$0<\lambda_1(Q)$, and
\begin{eqnarray}\nonumber
\dfrac{\lambda_{nN}(Q)}{\lambda_1(Q)} \geq 1 ~\Rightarrow~\dfrac{\lambda_{nN}(Q)}{\lambda_1(Q)} + 1 > \dfrac{\lambda_{nN}(Q)}{\lambda_1(Q)} - 1,
\end{eqnarray}
and the lemma follows.
\end{proof}
The above lemma provides a fundamental result that for any connected and (one-step) observable system, there exist unstable ($a>1$) dynamics that can be tracked by the estimator in~\eqref{est2} with bounded MSE\footnote{Note that, due to~\eqref{CalC}, the above lemma applies to the general class of estimators in~\eqref{est1}, i.e,~$C>1$ for connected (agent network) systems that are collectively observable in one-time step.}.

\subsection{Range of~$\alpha$}
In the following lemma, we show a range of~$\alpha$ that ensures a stable error process, i.e.,~$\|P\|_2<1,\forall~a<C_\alpha$.
\begin{lem}\label{lem_range}
Let~$a<C_\alpha$, then~$\|P\|_2<1$ for
\begin{eqnarray}\label{alp_int}
\alpha \in \left(\dfrac{a-1}{a\lambda_1(Q)},\dfrac{a+1}{a\lambda_{nN}(Q)}\right)\triangleq (\alpha_0,\alpha_1).
\end{eqnarray}
\end{lem}
\begin{proof}
Assume~$a<C_\alpha$, then from~\eqref{ntc_exp2} we have
\begin{eqnarray}\nonumber
\dfrac{1}{\lambda_1(Q)} - \dfrac{1}{a\lambda_1(Q)} < \dfrac{1}{\lambda_{nN}(Q)} + \dfrac{1}{a\lambda_{nN}(Q)}.
\end{eqnarray}
Hence, we can always choose an~$\alpha$ in the interval defined in~\eqref{alp_int}, i.e., because of the strict inequality in the above equation, we can always choose an~$\alpha$ such that
\begin{eqnarray}
\dfrac{1}{\lambda_1(Q)} - \dfrac{1}{a\lambda_1(Q)} <\alpha< \dfrac{1}{\lambda_{nN}(Q)} + \dfrac{1}{a\lambda_{nN}(Q)}.
\end{eqnarray}
The right inequality implies
\begin{eqnarray}\label{rhs_e1}
\alpha\lambda_{nN}(Q) - 1 < \dfrac{1}{a},
\end{eqnarray}
whereas the left inequality implies
\begin{eqnarray}\label{lhs_e1}
1 - \alpha\lambda_{1}(Q) < \dfrac{1}{a}.
\end{eqnarray}
Combining~\eqref{rhs_e1} and~\eqref{lhs_e1}, we note that there exists an~$\alpha$ in the interval defined in~\eqref{alp_int} such that
\begin{eqnarray}\nonumber
\max\{1 - \alpha\lambda_{1}(Q),\alpha\lambda_{nN}(Q) - 1\} < \dfrac{1}{a}.
\end{eqnarray}
Hence, we have
\begin{eqnarray}\nonumber
\|I-\alpha Q\|_2 = \max\{1 - \alpha\lambda_{1}(Q),\alpha\lambda_{nN}(Q) - 1\} < \dfrac{1}{a},
\end{eqnarray}
and the lemma follows. Furthermore, it can also be shown that~$\forall~\overline{\alpha}\geq\alpha_1$,~$\|I-\overline{\alpha} Q\|_2 = \overline{\alpha}\lambda_{nN} - 1 \geq 1/a$. Similarly,~$\forall~\underline{\alpha}\leq\alpha_1$, then~$\|I-\underline{\alpha} Q\|_2 = 1 - \underline{\alpha}\lambda_{1} \geq 1/a$. Hence, any~$\alpha\notin(\alpha_0,\alpha_1)$ does not guarantee~$\|P\|_2<1$.
\end{proof}
The above lemma provides an interval such that any choice of~$\alpha\in(\alpha_0,\alpha_1)$ guarantees a stable estimator, i.e.,~$\|P\|_2<1$, as long as~$a<C_\alpha$. We now show that that when~$a<C_\alpha$, then the~$\alpha_{\mbox{\scriptsize opt}}\in(\alpha_0,\alpha_1)$, as given in~\eqref{alpha_opt2}.
\begin{lem}
Let~$a<C_\alpha$, then~$\alpha_0<\alpha_{\mbox{\scriptsize opt}}<\alpha_1$.
\end{lem}
\begin{proof}
Since~$a<C_\alpha$, we have
\begin{eqnarray}
\dfrac{1}{C_\alpha} < \dfrac{1}{a},
\end{eqnarray}
which implies that
\begin{eqnarray}\label{lemc1}
\dfrac{C_\alpha+1}{C_\alpha} &<& \dfrac{a+1}{a},\\\label{lemc2}
\dfrac{C_\alpha-1}{C_\alpha} &>& \dfrac{a-1}{a}.
\end{eqnarray}
Now note that
\begin{eqnarray}\label{lemd1}
\dfrac{C_\alpha+1}{C_\alpha} &=& \dfrac{2\lambda_{nN}(Q)}{\lambda_{nN}(Q)+\lambda_{1}(Q)},\\\label{lemd2}
\dfrac{C_\alpha-1}{C_\alpha} &=& \dfrac{2\lambda_{1}(Q)}{\lambda_{nN}(Q)+\lambda_{1}(Q)}.
\end{eqnarray}
Combining~\eqref{lemc1} with~\eqref{lemd1}, and~\eqref{lemc2} with~\eqref{lemd2}, we get
\begin{eqnarray}
\dfrac{a-1}{a\lambda_{1}(Q)}<\dfrac{2}{\lambda_{nN}(Q)+\lambda_{1}(Q)}<\dfrac{a+1}{a\lambda_{nN}(Q)}
\end{eqnarray}
and the lemma follows.
\end{proof}
Further note that the length of the interval~$(\alpha_0,\alpha_1)$ is given by
\begin{eqnarray}\nonumber
\alpha_1-\alpha_0 &=& \dfrac{a+1}{a\lambda_{nN}(Q)} - \dfrac{a-1}{a\lambda_{1}(Q)},\\\nonumber
&=& \dfrac{\lambda_{1}(Q)+\lambda_{nN}(Q)-a(\lambda_{nN}(Q)-\lambda_{1}(Q))}{a\lambda_{1}(Q)\lambda_{nN}(Q)},\\\nonumber
&=&(\lambda_{nN}(Q)-\lambda_{1}(Q))\dfrac{\frac{\lambda_{1}(Q)+\lambda_{nN}(Q)}{\lambda_{nN}(Q)-\lambda_{1}(Q)}-a}{a\lambda_{1}(Q)\lambda_{nN}(Q)},\\
&=&\left(\dfrac{1}{\lambda_{1}(Q)}-\dfrac{1}{\lambda_{nN}(Q)}\right)\left(\dfrac{C_\alpha}{a} - 1\right),
\end{eqnarray}
confirming a non-empty interval for~$\alpha$ when~$a<C_\alpha$ or~$C_\alpha/a>1$. The interval length is larger when~$C_\alpha\gg a$, providing more choices for choosing an~$\alpha$ that give a stable estimator. On the other hand, the interval length is smaller when~$a\rightarrow C_\alpha$.

Note that the~$\alpha$ interval,~$(\alpha_0,\alpha_1)$, depends on~$a=\|A\|_2$ and the eigenvalues of~$Q$, which, in turn, depends on the graph Laplacian,~$L$, and all of the observation models. Hence, computing this interval to select a suitable~$\alpha$ requires global knowledge that may not be available at each agent. In Section~\ref{ld}, we provide an interval for~$\alpha$ that is a subset of~\eqref{alp_int} but can be computed from quantities that may be locally available.

\subsection{Performance of the scalar gain estimator}
In case of scalar gain estimators, we have
\begin{eqnarray}\nonumber
\|S_\infty\|_2 &\leq& \dfrac{N\|V\|_2 + \alpha^2\|(I_N\otimes A) (\mc{A}\otimes I_n)\overline{R}(\mc{A}\otimes I_n)^T(I_N\otimes A)^T\|_2}{1 - a^2\|I_{nN} - \alpha Q\|_2^2},\\\nonumber
&\leq& \dfrac{N\|V\|_2 + \alpha^2a^2\|\mc{A}\|_2^2\|\overline{R}\|_2}{1 - a^2\|I_{nN} - \alpha Q\|_2^2},\\
\Rightarrow\dfrac{1}{N}\|S_\infty\|_2&\leq& \dfrac{\|V\|_2 + \alpha^2a^2N\|\overline{R}\|_2}{1 - a^2\|I_{nN} - \alpha Q\|_2^2},
\end{eqnarray}
Clearly, a convex optimization similar to~\eqref{co_perf} may be formulated to design~$\alpha$ that results into a stable estimator with a performance guarantee.

{\bf Remarks:}
\begin{enumerate}[(i)]
\item We now show the performance bound with the optimal value of~$\alpha$ from~\eqref{alpha_opt2}. With~$\alpha = \alpha_{\mbox{\scriptsize opt}}$, we have~$\|I_{nN} - \alpha Q\|_2=1/C_\alpha$. The above upper bound can be simplified as
\begin{eqnarray}\nonumber
\dfrac{1}{N}\|S_\infty\|_2&\leq&  \dfrac{\|V\|_2 + (\frac{2}{C_\alpha(\lambda_{nN}(Q)+\lambda_{1}(Q))})^2a^2N\|\overline{R}\|_2}{1 - a^2(\frac{1}{C_\alpha})^2},\\\nonumber
&=& \dfrac{C_\alpha^2\|V\|_2 + \frac{4}{(\lambda_{nN}(Q)+\lambda_{1}(Q))^2}a^2N\|\overline{R}\|_2}{C_\alpha^2 - a^2},\\\nonumber
&=& \dfrac{\frac{C_\alpha^2}{a^2}\|V\|_2 + \frac{4}{(\lambda_{nN}(Q)+\lambda_{1}(Q))^2}N\|\overline{R}\|_2}{\frac{C_\alpha^2}{a^2} - 1},
\end{eqnarray}
Clearly, when~$a<C_\alpha$, the steady-state error remains bounded.

\item Consider the performance of neutrally-stable scalar systems (see Section~\ref{sc:nss}), i.e., we have~$a=1, n=1$, where~$Q=L+D_H$. The steady-state performance at each agent is
\begin{eqnarray}\nonumber
\dfrac{1}{N}\|S_\infty\|_2&\leq&  \dfrac{\sigma_v^2 + \alpha^2N\|\overline{R}\|_2}{1 - \|I_{n} - \alpha (L+D_H)\|_2^2},
\end{eqnarray}
where~$\|V\|_2\triangleq\sigma_v^2$, similar to the noise variance in the scalar dynamics,~$x_k$. It is straightforward to note that there exists~$\alpha$ such that~$0\leq\|I_{n} - \alpha Q\|_2^2<1$ (following Lemma~\ref{lem_sc:nss}) and the denominator is~$>0$. Following~\cite{4739167}, if we choose~$\alpha$ such that~$\alpha(k)\rightarrow 0$ (but sum to infinity) and we have diminishing innovation, i.e.,~$\sigma_v^2\rightarrow 0$, the steady-state error is bounded above by~$0$ (perfect learning).

\item Assume that the observation models at each agent is identical, i.e.,~$y_{i,k} = x_k + r_{i,k}$ and let~$\|r_k\|_2<\sigma^2_r,~\forall~k$, then
\begin{eqnarray}\nonumber
\dfrac{1}{N}\|S_\infty\|_2&\leq&  \dfrac{\sigma_v^2 + \alpha^2N\sigma_r^2}{1 - \|I_{n} - \alpha (L+I_N)\|_2^2}.
\end{eqnarray}
When we do not have diminishing innovation, i.e.,~$\sigma_v^2$ is fixed for all~$k$, then we can easily see that choosing an~$\alpha$ relates to the trade-off between the convergence and steady-state performance of the single time-scale estimator. Typically,~$\alpha\propto \frac{1}{\sqrt{N}}$ to remove the dependence of the steady-state error bound on the number of agents,~$N$.
\end{enumerate}

\section{Local design of~$\alpha$}\label{ld}
The discussion on NTC and scalar gain estimators requires centralized computation of the estimator parameters ($W,B$, or~$\alpha$ for scalar gain estimators). In this section, we present methodologies to obtain the scalar parameter,~$\alpha$, locally at each agent. For this purpose, we use Lemma~\ref{lem_range} that provides a range of~$\alpha\in(\alpha_0,\alpha_1)$ to implement stable estimators and consider strategies to compute this interval locally.

We first provide some eigenvalue bounds to facilitate the development in this section. We then consider~$m$-circulant graphs and special observation models to provide a local choice of~$\alpha$. Finally, we provide generalizations to the~$m$-circulant graphs and arbitrary observation models.

\subsection{Eigenvalue bounds}
To derive an interval of~$\alpha$ that may be locally computed, we use the following results from the matrix perturbation theory~\cite{stewart_book}. Let~$A_1$ and~$B_1$ be symmetric positive semi-definite matrices with eigenvalues~$0\leq\lambda_1(A_1)\leq\ldots\leq\lambda_{\max}(A_1)$, and~$0\leq\lambda_1(B_1)\leq\ldots\leq\lambda_{\max}(B_1)$, respectively. Then
\begin{subequations}\label{lem_ebnd}
\begin{eqnarray}
\lambda_1(A_1) &\leq& \lambda_1(A_1+B_1),\\
\lambda_{\max}(A_1+B_1) &\leq& \lambda_{\max}(A_1) + \lambda_{\max}(B_1).
\end{eqnarray}
\end{subequations}
In addition, we will use the following lemma from~\cite{ipsen:09}.

\begin{lem}[Theorem 2.1 in~\cite{ipsen:09}]\label{ipsen_th}
Let~$\mb{z}$ be an arbitrary column vector and let~$A_1$ be a symmetric positive definite matrix with eigenvalues,~$0\leq\lambda_1(A_1)\leq\ldots\lambda_{\max}(A_1),$ then
\begin{eqnarray}\nonumber
\lambda_1(A_1+\mb{z}\mb{z}^T) \geq \lambda_1(A_1) + \dfrac{1}{2}\left(\mbox{gap}_{1} + \|\mb{z}\|^2 - \sqrt{(\mbox{gap}_{1} + \|\mb{z}\|^2)^2 - 4\mbox{gap}_{1}|z_1|^2}\right),
\end{eqnarray}
where
\begin{eqnarray}\nonumber
\mbox{gap}_{1} &=& \lambda_2(A_1) - \lambda_1(A_1) \geq 0,\\\nonumber
z_1 &=& \mb{q}^\ast_1(A) \mb{z},
\end{eqnarray}
the column vector~$\mb{q}_1(A)$ is the eigenvector of~$A_1$ corresponding to the minimum eigenvalue~$\lambda_1(A_1)$ and~`$\ast$' denotes the Hermitian.
\end{lem}

Now consider~$A_1$ in the above lemma to be an~$N\times N$ graph Laplacian matrix,~$L$. Then the eigenvector~$\mb{q}_1(L)$ corresponding to the minimum eigenvalue~$\lambda_1(L)=0$ is
\begin{eqnarray}\nonumber
\mb{q}_1(L) = \left[\dfrac{1}{\sqrt{N}},~\ldots,\dfrac{1}{\sqrt{N}}\right]^T.
\end{eqnarray}
Lemma~\ref{ipsen_th} translates into the following lower bound on the minimum eigenvalue of~$L+\mb{z}\mb{z}^T$.
\begin{eqnarray}\nonumber
\lambda_1(L+\mb{z}\mb{z}^T) &\geq& \dfrac{1}{2}\left(\lambda_2(L) + \|\mb{z}\|^2 - \sqrt{(\lambda_2(L) + \|\mb{z}\|^2)^2 -4\lambda_2(L)|z_1|^2}\right),\\\nonumber&\triangleq& \tau(\lambda_2(L), \mb{z}).
\end{eqnarray}
Clearly, the above lower bound is only non-trivial ($\tau(\lambda_2(L), \mb{z})>0$) when~$\lambda_2(L)>0$, i.e., for connected graphs and is~$0$ for disconnected graphs ($\lambda_2(L)=0$). In particular, if we choose~$\mb{z}=\mb{z}_1=[1,~0,\ldots,~0]$, then
\begin{eqnarray}\nonumber
\tau(\lambda_2(L),\mb{z}_1)= \dfrac{1}{2}\left(\lambda_2(L) + 1 - \sqrt{(\lambda_2(L) + 1)^2 - \dfrac{4\lambda_2(L)}{N}}\right).
\end{eqnarray}

\subsection{Graphs isomorphic to~$m$-circulant graphs}
In this section, we consider a particular example and derive the eigenvalue bounds in that case using the results from the previous subsection. We choose an~$m$-circulant communication graph\footnote{An~$m$-circulant graph is a graph where the~$N$ nodes are arranged as distinct points on a circle and each node is connected to~$m$-forward neighbors.} with~$N$ nodes and an~$n$-dimensional dynamical system such that~$n=N$. We choose the following observation model at the~$i$th agent:
\begin{eqnarray}\label{sc_obs}
y_{k}^i = x^i_k + r_k^i,
\end{eqnarray}
i.e., the observation matrix at the~$i$th agent,~$H_i$, is an~$N$-dimensional row vector with~$1$ at the~$i$th location and zeros everywhere else. We partition the matrix~$D_H$ (see~\eqref{DH}) as~$\overline{D}_H + \underline{D}_H$, where
\begin{eqnarray}
\overline{D}_H = \mbox{blockdiag}[H_1^TH_1,\ldots,H_N^TH_N],
\end{eqnarray}\label{uDH}
and~$\underline{D}_H = D_H - \overline{D}_H$. From~\eqref{lem_ebnd}, we have
\begin{eqnarray}\label{lemkj}
\lambda_1(L\otimes I_n + \overline{D}_H + \underline{D}_H)  \geq \lambda_1(L\otimes I_n + \overline{D}_H).
\end{eqnarray}
We have the following lemma.
\begin{lem}\label{spec_lem}
For~$m$-circulant graphs and observation models of the form~\eqref{sc_obs}, we have
\begin{eqnarray}\nonumber
\mathcal{S}(L\otimes I_n + \overline{D}_H) = \mathcal{S}(I_n\otimes (L + H_i^TH_i) ),
\end{eqnarray}
for any~$i$, where~$\mathcal{S}$ denotes the spectrum (eigenvalues) of a matrix.
\end{lem}
\begin{proof}
We define an~$N\times N$ (stride) permutation matrix,~$T$, such that~$T(L\otimes I_n)T^T = I_n\otimes L.$ With this permutation matrix, it can be verified that
\begin{eqnarray}\nonumber
T(L\otimes I_n + \overline{D}_H)T^T &=& I_n\otimes L + \overline{D}_H,
\end{eqnarray}
since~$\overline{D}_H$ is diagonal and consists of~$1$'s only at those diagonal locations that are left unchanged with the permutation,~$T$. The matrix~$I_n\otimes L + \overline{D}_H$ consists of~$N$~$n\times n$ blocks where the~$j$th block is given by~$L+H_j^TH_j$. The matrix~$H_j^TH_j$ is an~$n\times n$ matrix with a~$1$ at the~$(j,j)$ location and zeros everywhere else. Hence,~$L+H_j^TH_j$ is a Laplacian matrix whose~$(j,j)$ diagonal element is perturbed by~$1$. Since~$L$ corresponds to the Laplacian matrix of an~$m$-circulant graph, the spectrum is left unchanged regardless of which diagonal element is perturbed, i.e.,
\begin{eqnarray}\nonumber
\mathcal{S}(L+H_j^TH_j) = \mathcal{S}(L+H_i^TH_i), \qquad \forall~j\neq i.
\end{eqnarray}
Noting that the spectrum of a matrix does not change under a similarity transformation with a permutation matrix, we have for any~$i$
\begin{eqnarray}\nonumber
\mathcal{S}(L\otimes I_n + \overline{D}_H) = \mathcal{S}(T(L\otimes I_n + \overline{D}_H)T)
= \mathcal{S}(I_n\otimes L  + \overline{D}_H) = \mathcal{S}(I_n\otimes (L  + H_i^TH_i)).
\end{eqnarray}
\end{proof}

With the above lemma,~\eqref{lemkj} is further bounded below by
\begin{eqnarray}\nonumber
\lambda_1(L\otimes I_n + D_H)  &\geq& \lambda_1(I_n\otimes (L  + H_i^TH_i)),\\\nonumber
&=& \lambda_1(L  + H_i^TH_i),\\\label{great_lb}
&\geq& \tau(\lambda_2(L),H_i^T),
\end{eqnarray}
from~\eqref{ipsen_th}. Furthermore, for the observation models given by~\eqref{sc_obs}, the matrix~$D_H$ is diagonal whose~$(i,i)$ element is either~$0$ or~$1$. Hence, from~\eqref{lem_ebnd},
\begin{eqnarray}\nonumber
\lambda_{nN}(L\otimes I_n + D_H) &\leq& \lambda_{nN}(L\otimes I_n) + \lambda_{nN}(D_H),\\\label{great_ub}
&=&\lambda_{nN}(L\otimes I_n) + 1.
\end{eqnarray}

The discussion in this subsection holds true for any graph that is isomorphic to an~$m$-circulant graph. This is because a vertex relabeling of such a graph results into an~$m$-circulant graph. Hence, for any graph that is isomorphic to an~$m$-circulant graph,~\eqref{great_lb} serves as a lower bound of~$\lambda_1(L\otimes I_n + D_H)$ and~\eqref{great_ub} serves as an upper bound for~$\lambda_{nN}(L\otimes I_n + D_H)$. Thus, we have for any graph that is isomorphic to an~$m$-circulant graph and observation models of the form~\eqref{sc_obs},
\begin{eqnarray}\label{great_ulb}
\tau(\lambda_2(L),\mb{z}_1) \leq \lambda_1(Q) \leq \lambda_{nN}(Q) \leq 1 + \lambda_{nN}(L).
\end{eqnarray}
The following lemma now provides the main result of this section.
\begin{lem}
Consider a system with observation models in~\eqref{sc_obs} and a communication graph that is isomorphic to~$m$-circulant graphs. Let~$a<C_\alpha$, then~$\|P\|_2<1$ for
\begin{eqnarray}\label{alp_newint}
\alpha \in \left(\dfrac{a-1}{a\tau(\lambda_2(L),\mb{z}_1)},\dfrac{a+1}{a(1 + \lambda_{nN}(L))}\right).
\end{eqnarray}
\end{lem}
\begin{proof}
Because of~\eqref{great_ulb}, the interval defined in~\eqref{alp_newint} is a subset of the interval defined in~\eqref{alp_int} and the lemma follows form Lemma~\ref{lem_range}.
\end{proof}

\subsection{Generalization}
In this section, we generalize the development in the previous subsection to arbitrary graphs that contain at least one simple cycle of length~$N$ as a subgraph. We denote such a graph as~$\overline{\mathbb{G}_{\circlearrowleft N}}=(\mathcal{V},\overline{\mathcal{E}_{\circlearrowleft N}})$ and denote its Laplacian matrix as~$\overline{L_{\circlearrowleft N}}$. Let~$\mathbb{G}_{\circlearrowleft N}=(\mathcal{V},{\mathcal{E}_{\circlearrowleft N}})$ denote a simple cycle of length~$N$ and let~$L_{\circlearrowleft N}$ denote its Laplacian matrix. Then,~$\overline{L_{\circlearrowleft N}} = L_{\circlearrowleft N} + L_1$, where~$L_1$ corresponds to the Laplacian matrix of~$(\mathcal{V},\overline{\mathcal{E}_{\circlearrowleft N}}\setminus\mathcal{E}_{\circlearrowleft N})$. We assume that the observation models are such that the matrix~$D_H$ can be decomposed into~$\overline{D}_H$ (see~\eqref{uDH}) plus some other~${D}_{H1}$, where~$D_{H1}$ is a positive semi-definite matrix.
From~\eqref{lem_ebnd}, we have
\begin{eqnarray}\nonumber
\lambda_1(\overline{L_{\circlearrowleft N}}\otimes I_n + D_H) \geq \lambda_1(L_{\circlearrowleft N}\otimes I_n + \overline{D}_H).
\end{eqnarray}
Hence, for all graphs~$\overline{\mathbb{G}_{\circlearrowleft N}}$ and observation models that can be decomposed as~$\overline{D}_H+{D}_{H1}$ with~${D}_{H1}$ being positive semi definite, any
\begin{eqnarray}\label{alp_newint2}
\alpha \in \left(\dfrac{a-1}{a\tau(\lambda_2(L),\mb{z}_1)},\dfrac{a+1}{a(\lambda_{nN}(D_H) + \lambda_{nN}(L))}\right)
\end{eqnarray}
results into stable error processes. Clearly, the above interval is only non-empty when~$a<C_{\mbox{\scriptsize loc}}$, where
\begin{eqnarray}
C_{\mbox{\scriptsize loc}} = \dfrac{\lambda_{nN}(D_H) + \lambda_{nN}(L) + \tau(\lambda_2(L),\mb{z}_1)}{\lambda_{nN}(D_H) + \lambda_{nN}(L) - \tau(\lambda_2(L),\mb{z}_1)}.
\end{eqnarray}
It can be shown that~$C_{\mbox{\scriptsize loc}}\leq C_\alpha$. The gap between the two capacities (local,$C_{\mbox{\scriptsize loc}}$ and scalar optimal,~$C_\alpha$) depends on the tightness of the eigenvalue bounds used and on the interval length (proportional to~$C_\alpha/1 -1~$).

For most classes of structured graphs, the minimum and maximum eigenvalues are known in closed-form as a function of~$N$ and hence, can be used at each agent to compute the interval for~$\alpha$ in~\eqref{alp_int}. In the case of non-structured graphs, one may use properties of a given graph to bound the minimum and maximum eigenvalues of its Laplacian and use the interval in~\eqref{alp_newint2}. Similarly, for structured observation models,~$\lambda_{nN}(D_H)$ may be computed or upper-bounded (in the worst case) at the agents to locally compute the interval for~$\alpha$. Clearly, when we use a tighter interval for~$\alpha$ using a local procedure, we pay a price in terms of a loss in capacity.

\section{Illustration}\label{exam}
In this section, we provide an illustration of the concepts introduced in this paper. Consider an~$N$-node~$m$-circulant graph,~$\mathbb{G}_{\circlearrowleft^m}$, i.e., the nodes are arranged as distinct points on a circle and each node is connected to its next~$m$ neighbors. For such a graph, the Laplacian matrix,~$L_{\circlearrowleft^m}$, is a circulant matrix that can be diagonalized by an~$N\times N$ normalized DFT matrix. We consider~$M=1,2,3$, as shown in Fig.~\ref{mgr}.
\begin{figure}
\centering
\includegraphics[width=5in]{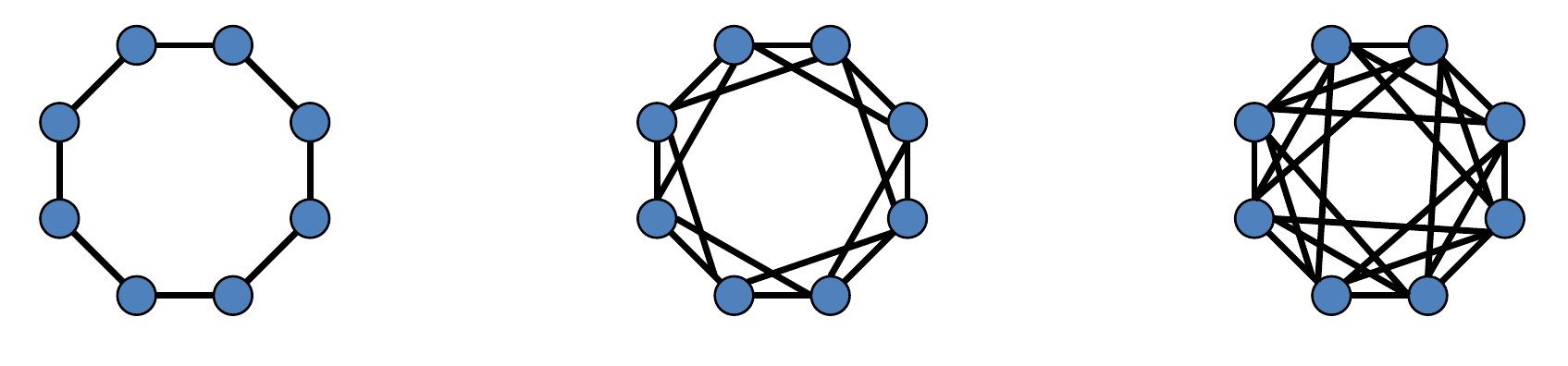}
\caption{Circulant graphs with~$m=1,2,3$.}
\label{mgr}
\end{figure}
We further consider an~$n=N$ dimensional state-space and scalar observation models such that the~$i$th node observes a noisy version of~$x_k^i$ at time~$k$, as given in~\eqref{sc_obs}. The capacity,~$C$, of the graphs in Fig.~\ref{mgr} is shown in Fig.~\ref{simfig_ntc_circ}(left). The capacity is calculated by using a convex program to solve~\eqref{ntc_ex}. Note that~$m=1, N=2,3$,~$m=2,N=2,\ldots,5$, and~$m=3,N=2,\ldots,7$ are fully-connected graphs. For fully-connected graphs, the capacity is infinite if we use the estimator in~\eqref{est1} as discussed after Lemma~\ref{lem_ntc_ran}.

We now consider scalar gain estimators. The capacity,~$C_\alpha$, of the~$m$-circulant graphs with scalar observation models is plotted in Fig.~\ref{simfig_ntc_circ}(right) for~$m=1,\ldots,3$ as a function of~$N$. If we restrict ourselves to scalar gain estimators then full capacity is not achievable with the scalar observation models even with the fully-connected graphs.
\begin{figure}
\centering
\subfigure{\includegraphics[width=3in]{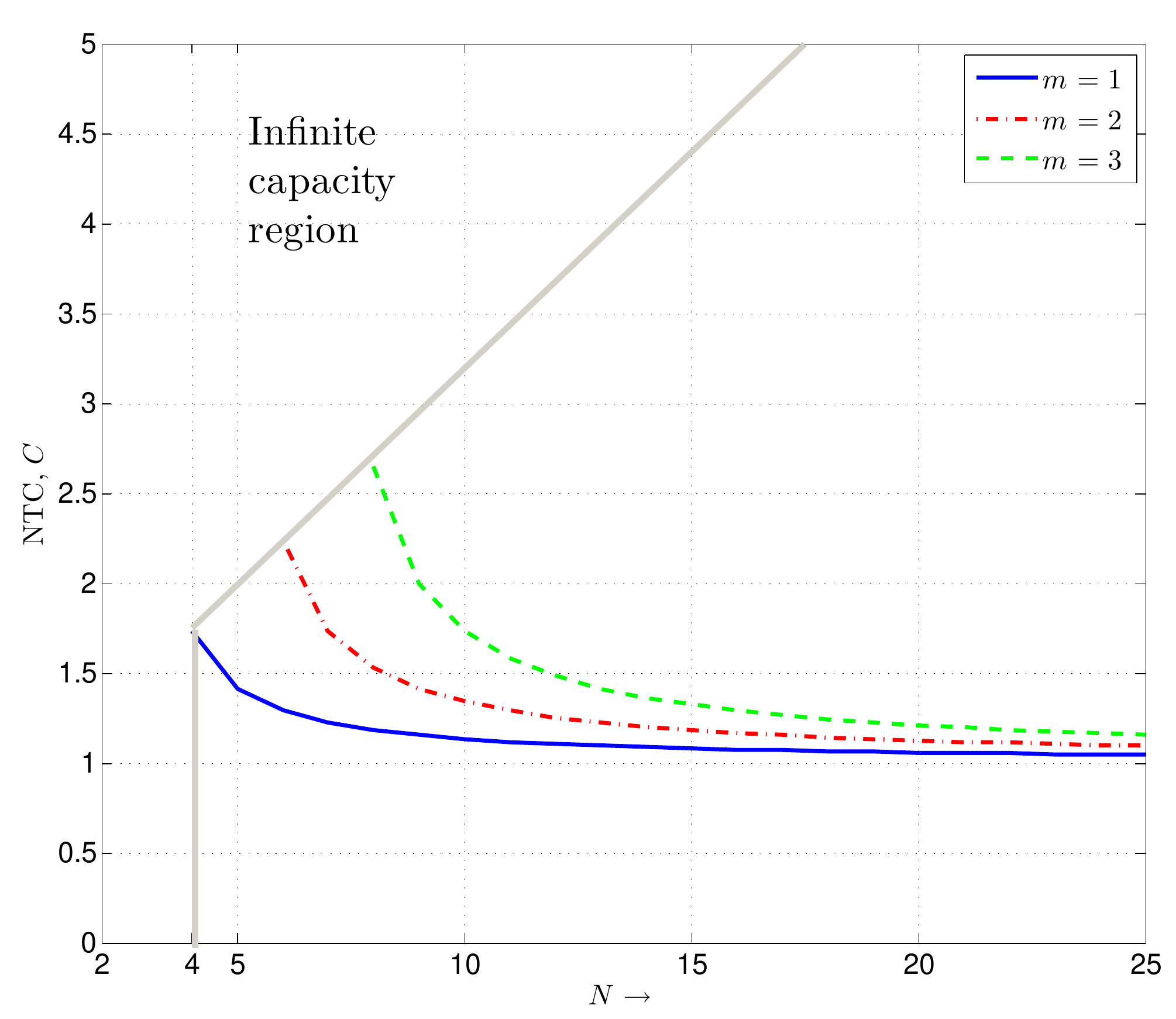}}
\hspace{1cm}
\subfigure{\includegraphics[width=3in]{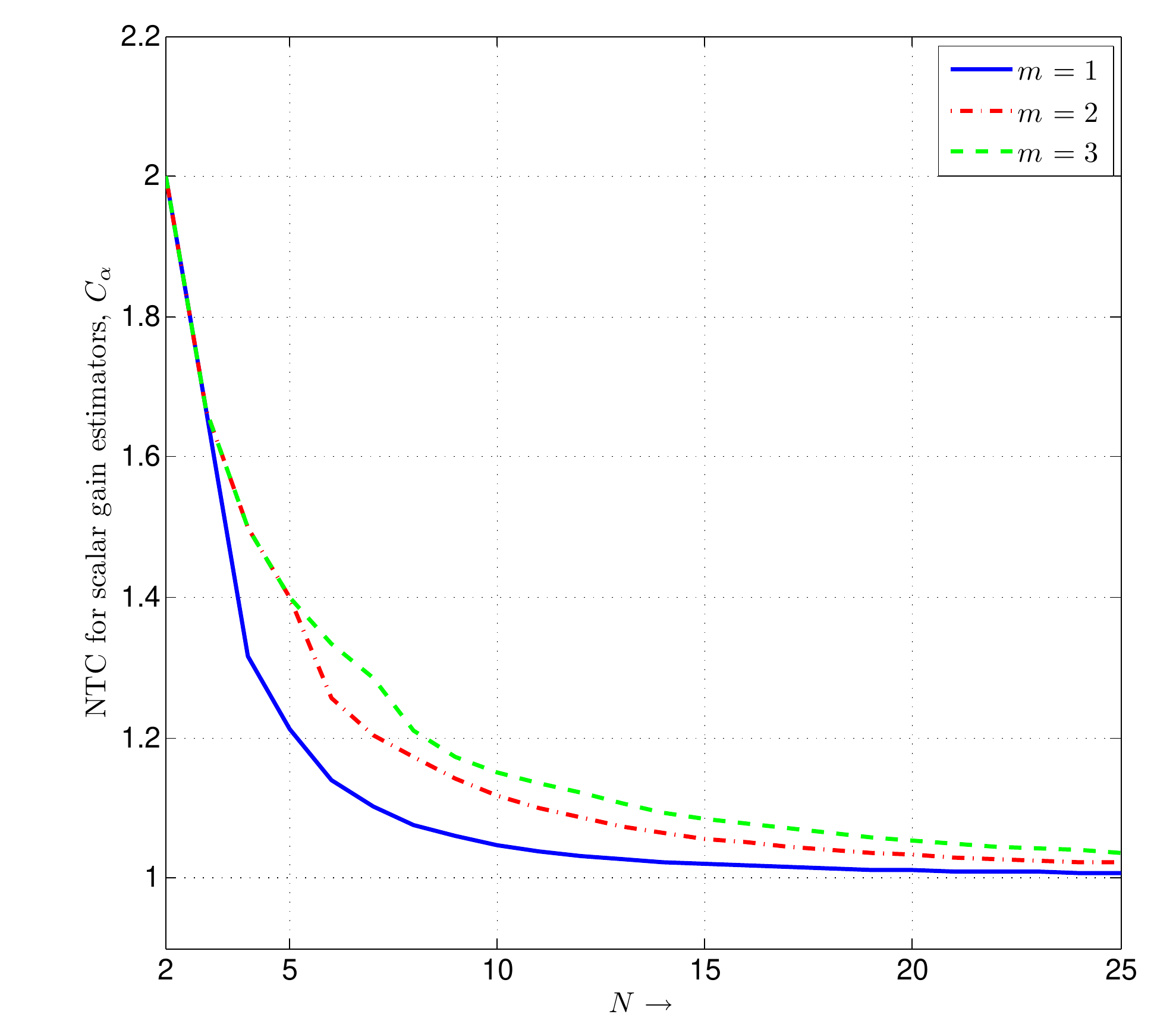}}
\caption{(Left) NTC,~$C$ for~$m$-circulant graphs with scalar observation models. Notice for fully connected graphs the capacity is infinite. (Right) NTC,~$C_\alpha$, for scalar gain estimator and~$m$-circulant graphs with scalar observation models.}
\label{simfig_ntc_circ}
\end{figure}
It can also be noted that the NTC is always greater than unity.

Figs.~\ref{simfig_ntc_circ} explicitly characterize the relation of estimable systems to the underlying (agent) network connectivity. As we increase~$m$ in a circulant graph (with fixed number of nodes,~$N$), the information flow (or the algebraic connectivity) increases resulting into a richer set of estimable systems. The results explicitly show that the system (in)stability plays a vital roles in the estimation capacity of a network. For a weakly-connected network, we pay a price by being restricted to only those dynamical systems that evolve slower.

\subsection{Each agent observable}
Consider each agent~$i$ to be such that~$H_i^TH_i$ is invertible~$\forall~i$. In this case, choose
\begin{eqnarray}
W &=& 0_{nN},\\
B &=& \mbox{blockdiag}[(\sum_{j\in\mc{N}_1}H_j^TH_j)^{-1},\ldots,(\sum_{j\in\mc{N}_N}H_j^TH_j)^{-1}].
\end{eqnarray}
With the above~$W$ and~$B$, we have
\begin{eqnarray}
W - BD_H &=& \mb{0}_{nN},
\end{eqnarray}
and~$C=\infty$, i.e., any dynamical system ($a<\infty$) can be tracked with bounded error. In other words, when each agent is observable then any arbitrary dynamical system can be tracked with bounded error regardless of the network.

\subsection{No observations}
Consider the system to have no observations, i.e.,~$D_H=\mb{0}$. In this case, no matter we choose for~$W$ and~$B$, we will have NTC,~$C=1$ and only stable dynamical system ($a<C_\alpha=1$) can be tracked with bounded MSE. This example shows that if there are no observations, only stable dynamical system can be tracked with bounded error. This is intuitive as any consistent estimate of a stable system results into bounded MSE.

\section{Conclusions}\label{conc}
In this paper, we explore two networked estimator design paradigms based on LMI methods and two-norm relaxations. The premise of our approach is single time-scale algorithms where only one informatione xchange is allowed between each successive system evolution step. We particularly consider vector state-space where the observability is assumed on the collection of all of the agent measurements, i.e., any strict subset of agents may not be necessarily observable. Our formulation considers arbitrary dynamics and is not restricted to stable or neutrally-stable systems.

We first consider the spectral design using a cone cone complementarity linearization algorithm that requires minimal assumptions but may not result in a solution. We then resort to a two-norm relaxation of the spectral radius and provide the Network Tracking Capacity. We show that the proposed networked estimator results into a bounded MSE for all linear dynamical systems whose instability (in the~$2$-norm sense) is strictly less than the NTC. For both procedure, we explicitly provide non-trivial upper bounds on the steady-state covariance and further explore convex procedures that address performance in addition to stability. We then consider a simple networked estimator where the design is restricted to a single parameter that we term as scalar-gain estimators. With the help of scalar-gain estimators, we provide completely local design principles to implement local estimators with bounded MSE.

\appendices

\bibliographystyle{IEEEbib}
\bibliography{bibliography}

\end{document}